\documentclass[final,5p,times,twocolumn]{elsarticle}

\usepackage{latexsym}

\usepackage{amsmath}
\usepackage{amsthm}
\usepackage{amssymb}
\usepackage{enumerate}
\usepackage[normalem]{ulem}
\usepackage{verbatim}

\usepackage{subfigure}
\usepackage{tikz}

\tikzstyle{place}=[circle,draw=black,fill=gray!15,thick,inner sep=0pt,minimum size=6mm]

\usetikzlibrary{arrows}

\usetikzlibrary{snakes}
\usepackage{colortbl}
\usepackage{tikz}
\tikzstyle{place}=[circle,draw=black,fill=gray!15,thick,inner sep=0pt,minimum size=5mm]
\tikzstyle{place1}=[circle,draw=black,fill=black!95,thick,inner sep=0pt,minimum size=1mm]
\tikzstyle{place3}=[circle,draw=black,fill=black,thick,inner sep=.3pt]
\usepackage{enumerate}
\usetikzlibrary{arrows}
\usetikzlibrary{decorations.pathmorphing}

\newtheorem{theorem}{Theorem} 
 
\newtheorem{lemma}{Lemma}

\newcommand{\rr}{\ensuremath{\mathbb{R}}}

\newcommand{\val}{\ensuremath{\text{Val}}}

\newcommand{\arval}{\ensuremath{\text{AVal}}}

\newcommand{\ca}{\ensuremath{\text{Cap}}}

\newcommand{\opt}{\textrm{Z}}

\usepackage{enumerate}
\usepackage{color}

\usepackage{subfigure}

\newcommand{\prob}[3]%
  { \vspace{.25cm}
    \noindent\fbox{\begin{minipage}{.98\textwidth}
      \textsc{#1}

      \vspace{.25cm}
      \begin{center}
        \begin{tabular}{lp{.8\textwidth}}
          Input:  & #2 \\
          Task: & #3
        \end{tabular}
      \end{center}
    \end{minipage}}

    \vspace{.25cm}}%

\newcommand{\probbox}[3]%
  { \vspace{.5cm}
    \noindent{\begin{minipage}{.98\textwidth}
      \textsc{\hspace{.5cm}#1}

      \vspace{.25cm}
      \begin{center}
        \begin{tabular}{lp{.8\textwidth}}
          Given:  & #2 \\
          Question: & #3
        \end{tabular}
      \end{center}
    \end{minipage}}

    \vspace{.25cm}}%

  { \vspace{.25cm}
    \noindent\fbox{\begin{minipage}{.98\textwidth}
      \textsc{#1}

      \vspace{.25cm}
      \begin{center}
        \begin{tabular}{lp{.8\textwidth}}
          Input:  & #2 \\
          Output: & #3
        \end{tabular}
      \end{center}
   \end{minipage}}
          \vspace{.25cm}
    \begin{enumerate}[(1)]}%
  {\end{enumerate}}

\journal{...}
\begin{document}

\begin{frontmatter}
\date{}



\title{On the power of randomization in network interdiction}


\author{Dimitris~Bertsimas, Ebrahim Nasrabadi, James~B.~Orlin}
\address{Sloan School of Management and Operations Research Center, Bldg. E40-147, Massachusetts Institute of Technology, Cambridge, Massachusetts 02139\\
E-mail: \{dbertsim,nasrabad,jorlin\}$@$mit.edu}

\begin{abstract}
Network interdiction can be viewed as a game between two players, an \emph{interdictor} and a \emph{flow player}. The flow player wishes to send as much material as possible through a network, while the interdictor attempts to minimize the amount of transported material by removing a certain number of arcs, say $\Gamma$ arcs. We introduce the \emph{randomized network interdiction} problem that allows the interdictor to use randomness to select arcs to be removed. We model the problem in two different ways: \emph{arc-based} and \emph{path-based} formulations, depending on whether flows are defined on arcs or paths, respectively. We present insights into the modeling power, complexity, and approximability of both formulations.
 In particular, we prove that $\opt_{\text{NI}}/\opt_{\text{RNI}}\leq \Gamma+1$, $\opt_{\text{NI}}/\opt_{\text{RNI}}^{\text{Path}}\leq \Gamma+1$,  $\opt_{\text{RNI}}/\opt_{\text{RNI}}^{\text{Path}}\leq \Gamma$, where $\opt_{\text{NI}}$, $\opt_{\text{RNI}}$, and $\opt_{\text{RNI}}^{\text{Path}}$  are the optimal values of the network interdiction problem and its randomized versions in arc-based and path-based formulations, respectively. We also show that these bounds are tight. We show that it is NP-hard to compute the values $\opt_{\text{RNI}}$ and $\opt_{\text{RNI}}^{\text{Path}}$ for a general $\Gamma$, but they are computable in polynomial time when $\Gamma=1$. Further, we provide a $(\Gamma+1)$-approximation for $\opt_{\text{NI}}$, a $\Gamma$-approximation for $\opt_{\text{RNI}}$, and a $\big(1+\lfloor \Gamma/2\rfloor \cdot \lceil \Gamma/2\rceil/(\Gamma+1)\big)$-approximation for $\opt_{\text{RNI}}^{\text{Path}}$.
 \end{abstract}

\begin{keyword}
network flows\sep interdiction\sep game theory\sep  approximation

\end{keyword}

\end{frontmatter}


\section{Introduction}
Network flows have applications in a wide variety of contexts
(see, e.g., \cite{AhujaMagOrl93}). In some applications, it is useful to consider the perspective of someone who wants to restrict flows in a network. For example, law enforcement wants to inhibit the flow of illegal drugs. Water management experts want to control flows to avoid floods. Health agencies need to protect against contagion.  Here, it is important to consider the problem of limiting flows in the network from the perspective of an \emph{interdictor}, who is capable of limiting capacity in arcs or eliminating arcs.  Such problems have been applied in many application areas such as military planning \cite{Whiteman1999}, controlling infections in a hospital \cite{Assimakopoulos87}, controlling floods \cite{RatliffSiciliaLubore75}, protecting critical infrastructures \cite{MurrayMatisziwEtAl07,SalmeronWoodBaldick04},  and drug interdiction~\cite{Steinrauf91}.

Motivated by the above mentioned applications, \emph{network interdiction problems} have been well studied in the literature (see, e.g., \cite{BayrakBailey08,CormicanMortonWood98,IsraeliWood02,JanjarassukLinderoth08,LimSmith07,SanseverinoMarquez10,RoysetWood07,Wood93}).
In this paper, we focus on the basic model of network interdiction: a \emph{flow player} attempts to maximize the amount of materiel transported through a capacitated network, while an \emph{interdictor} tries to limit the flow player's achievable value by interdicting a certain number, say $\Gamma$, of arcs. This problem is also known as the $\Gamma$-most vital arcs problem~ (see, e.g., \cite{RatliffSiciliaLubore75}). Wollmer~\cite{Wollmer64} presents a polynomial time algorithm for solving this problem on planar graphs. On general networks, Wood~\cite{Wood93} shows that the problem is strongly NP-complete.

Network interdiction can be viewed as a game between the interdictor and the flow player. This problem assumes the interdictor moves first and then the flow player determines a maximum flow in the remaining network. 
A closely related problem arises when a flow must be routed before arcs are removed. In this case, the flow player might be interested to find solutions which are robust again any failure of arcs. 
Aneja  \emph{et al.} \cite{AnejaNairChand01} and Du and Chandrasekaran \cite{DuChan07} address this issue in a path-based formulation. They show that the resulting problem is solvable in polynomial-time for the special case of $\Gamma=1$, but becomes NP-hard if $\Gamma=2$. This problem was further expanded to an arc-based formulation by Bertsimas \emph{et al.}~\cite{BertNasrStil12}, who introduce the concepts of robust and adaptive maximum flows. They establish structural and computational results for both the robust and adaptive maximum flow problems and their corresponding minimum cut problems.

\paragraph{Our contribution} The network interdiction problem addresses a minimax objective against a flow player, which selects adaptively a flow after observing the removed arcs.  This problem requires the interdictor to choose a specific \emph{pure strategy}. We propose a new modeling framework that permits the interdictor to use randomness to choose arcs. More precisely, the interdictor assigns a probability to each pure strategy and selects a pure strategy randomly according to these probabilities.  We refer to the resulting problem as the \emph{randomized} network interdiction problem. This provides a  more realistic model for various applications such as protecting critical infrastructures against terrorism or enemy's attacks. We also consider a further modification that requires the flow player to send flow on paths, rather than the more typical arc-based  model. We present results on the modeling power, complexity, and approximability of both arc-based and path-based formulations.

More specifically, our primary contributions are as follows: 
\begin{enumerate}
\item \emph{\textbf{Modeling:}} We introduce the randomized network interdiction problem in arc-based and path-based formulations. We show that the arc-based (path-based) model is equivalent to the case where the flow player must determine an arc-based (path-based) flow through the network in advance and then the interdictor selects arcs to be removed. This shows that randomization helps the interdictor to perform as well as when she has perfect information of how flow is routed through the network.  


\item \emph{\textbf{Complexity:}} We present complexity results for both arc-based and path-based formulations of the randomized network interdiction problem. In particular, the arc-based model is NP-hard for a general $\Gamma$, but is solvable in polynomial time for a fixed $\Gamma$, while the path-based version is NP-hard for any fixed $\Gamma\geq 2$. For $\Gamma=1$, we show that both formulations become equivalent and are solvable in polynomial time as a linear optimization problem. 

\item \emph{\textbf{Approximability:}} 
We introduce a linear optimization problem and show that its optimal value provides a tight $(\Gamma+1)$--approximation, $\Gamma$--approximation, and $\big(1+\frac{\lfloor \Gamma/2\rfloor \cdot \lceil \Gamma/2\rceil}{\Gamma+1}\big)$--approximation for the optimal values of the network interdiction problem and its randomized versions in arc-based and path-based formulations, respectively. The latter approximation guarantee is $\frac{(\Gamma+2)^2}{4(\Gamma+1)}$ for even $\Gamma$ and is $\frac{\Gamma+3}{4}$ for odd $\Gamma$.  We note that for the $\Gamma=2$ case, it guarantees a $4/3$-approximation for the path-based problem, which is NP-hard. 

\item \emph{\textbf{Power of randomization:}}  
We show that the interdictor can perform significantly better by using randomization. In particular, we show that the ratio of the optimal value of network interdiction to that of the randomized version is bounded by $\Gamma+1$ for both arc-based and path-based models.  We also show that the ratio of the optimal value of the arc-based randomized network interdiction problem to that of the path-based version is bounded by $\Gamma$. We provide examples to show that these bounds are tight. 

\end{enumerate}

\section{Network Interdiction as a Game}
Let  $G=(V,E)$ be a directed graph with \emph{node set} $V$ and
\emph{arc set}~$E$. 
Each arc $e\in E$ has a \emph{capacity}~$u_{e}\in\rr_{+}$ setting an upper bound on the amount of flow on arc $e$. There are two specific nodes, a
{\em source} $s$ and a {\em sink}~$t$. W denote an arc~$e$ from a node~$v$ to a node~$w$ by~$e:=(v,w)$.
We use $\delta^+(v):=\{(v,w)\in E\mid w\in V\}$ and $\delta^-(v):=\{(w,v)\in E\mid w\in V\}$ to denote the sets of arcs leaving
node $v$ and entering node~$v$, respectively. We assume without loss of generality that there are no arcs into $s$ and no arcs out of $t$, that is, $\delta^-(s)=\delta^+(t)=\emptyset$. 

An $s$-$t$-\emph{flow} (or simply a \emph{flow}) $x$ is a function $x:E\rightarrow \rr_{+}$ which assigns a nonnegative value to each arc so that $x_e\leq u_e$ for each $e\in E$, and in addition for each node $v\in V\setminus\{s,t\}$, the following \emph{flow conservation constraint} holds:
\begin{align*}
\sum_{e\in \delta^-(v)}x_e-\sum_{e\in \delta^+(v)}x_e =0.
\end{align*}
We refer to $x_e$ as the \emph{flow} on arc $e$. We denote the set of all $s$-$t$-flows by $\mathcal{X}$. The \emph{value} $\val(x)$ of an $s$-$t$ flow $x$ is the net flow into $t$, that is, $ \val(x):= \sum_{e\in \delta^-(t)}x_e$.
In the \emph{maximum flow} problem (also referred to as the \emph{nominal} problem), we seek an $s$-$t$ flow $x$ with maximum
value $\val(x)$.


We next assume that there is an interdictor, who wants to reduce the capacity of the network. Suppose that the interdictor is able to eliminate $\Gamma$ ($1\leq \Gamma \leq |E|$) arcs in the network. 
The \emph{network interdiction} problem is to find the $\Gamma$ arcs
whose removal from the network minimizes the maximum amount of flow that can be sent to the sink.  To formulate this problem, we let
\begin{align*}
\Omega:=\Big\{\mu=(\mu_e)_{e\in E}\in \{0,1\}^{|E|}\mid  \sum_{e\in E}\mu_e= \Gamma\Big\}
\end{align*}
denote the set of all possible scenarios, that is, the set of all subsets of $\Gamma$ arcs. The binary variable $\mu_e$ indicates whether or not arc $e$ is to be removed, depending on whether $\mu_e=1$ or $\mu_e=0$, respectively.  
Given $\mu\in \Omega$, we denote by $E(\mu):=\{e\in E\mid \mu_e=1\}$ the set of removed arcs and by $F(\mu):=\{e\in E\mid \mu_e=0\}$ the set of available arcs after removing the arcs in the scenario $\mu$. We also denote by
 $G(\mu)=(V,F(\mu))$ a network with arc set $F(\mu)$.  

The network interdiction problem is formulated as
\begin{align}
  \label{pro:NI-org}
   \begin{aligned}
\min_{\mu\in \Omega}~&\max &&\quad \val(x) \\
         &     \text{~s.t.} &&\begin{aligned}[t]
        							x  &\in \mathcal{X},                &&\\
		  												            x_e                &=0&& \forall e\in E(\mu).
       													\end{aligned}
\end{aligned}
\end{align}
%
This problem can be viewed as a two-person zero-sum game between an interdictor and a flow player . The set of (pure) strategies for the interdictor is given by the scenario set $\Omega$.
The set of (pure) strategies for the flow player is given by the feasible set $\mathcal{X}$. With respect to a flow $x$ and scenario $\mu$, we denote by $G(\mu,x)$ a network with arc set $F(\mu)$ and arc capacities $x$. If the interdictor chooses the pure strategy $\mu\in \Omega$ and the flow player chooses the pure strategy $x\in \mathcal{X}$, then the payoff $f(\mu,x)$ of the game is the maximum amount of flow that the flow player can push through the network with respect to the flow $x$ if scenario $\mu$ is selected. Mathematically, the payoff function $f$ is given by
\begin{align}
\label{pro:f(mu,x)}
  \begin{aligned}
f(\mu,x):=&\max	&&\val(y) \\
         &     \text{~s.t.} &&\begin{aligned}[t]
        							y  &\in \mathcal{X},                &&\\
		  					0\leq y_e               &\leq x_e &&  \forall e\in F(\mu)\\
						            y_e                &=0&& \forall e\in E(\mu).
       													\end{aligned}
\end{aligned}
\end{align}

The interdictor aims to minimize the payoff of the game and must choose her strategy first. This provides an alternative and equivalent formulation of network interdiction as follows:
\begin{align}
  \label{pro:NI}
\opt_{\text{NI}}=\min_{\mu\in \Omega}~\max_{x\in \mathcal{X}} &\quad f(\mu,x).
\end{align}
This problem determines the interdictor's best choice, assuming the
flow player is in a position to select a maximum flow after observing the removed arcs. In many applications, the flow player have to make a decision before the interdictor selects her strategy. Here, the flow player might be interested in those solutions that are robust against any possible
scenario.  
This leads to the following problem, referred to as the \emph{adaptive maximum flow} problem:
\begin{align}
  \label{pro:Adp}
\opt_{\text{ADP}}:=\max_{x\in \mathcal{X}}~\min_{\mu\in \Omega} &\quad f(\mu,x).
\end{align}
This problem is introduced by Bertsimas \emph{et al.}~\cite{BertNasrStil12}, who establish structural properties and  complexity results for the problem. In particular, they show that the adaptive maximum flow problem is NP-hard using a reduction from the network interdiction problem.

Note that $\opt_{\text{ADP}}\leq \opt_{\text{NI}}$ and the equality may not be attained in general.  To compare the difference between the network interdiction problem and the adaptive maximum flow problem, we consider a network with three nodes $s$, $v$, and $t$ as shown in Figure \eqref{fig:NI-RNI-RNI(P)}. There are $K$ arcs with unit capacity and one arc with capacity $3/2K$ from $s$ to $t$ and there are $\Gamma+1$ arcs with infinite capacity from $v$ to $t$. Let $\Gamma\geq 2$ and $K$ be enough large. It is easy to see that $\opt_{\text{NI}}=K-1$, while $\opt_{\text{ADP}}=\frac{5K}{2(\Gamma+1)}$. Hence, $\opt_{\text{NI}}/\opt_{\text{ADP}}= \frac{2(\Gamma+1)(K-1)}{5K}$, and the ratio the becomes close to $2(\Gamma+1)/5$ when $K$ gets large. An interesting question is: How large can $\opt_{\text{NI}}/\opt_{\text{ADP}}$ be in general? We will show later that this ratio is bounded by $\Gamma+1$ and this bound is tight.

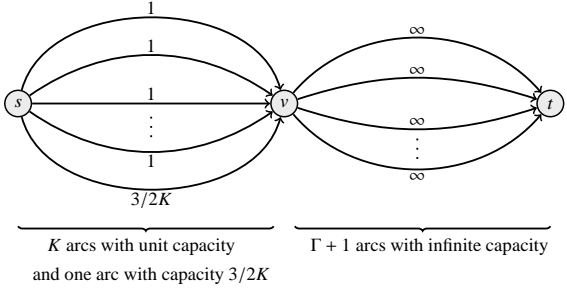
\begin{figure}[t]
\begin{center}
{\scalebox{0.7}
{\begin{tikzpicture}[inner sep=0.4mm]
  \node (1) at (0,0) [place] {$s$};
  \node (2) at (5,0) [place] {$v$};
  \node (3) at (10,0) [place] {$t$};
\begin{scope}[color=black,line width=1pt]
  \draw [->] (1) to [bend left=75] node [above, sloped,midway]{1}  (2) ;
   \draw [->] (1) to [bend left=35]   node [above, sloped,midway]{1}  (2) ;
 \draw [->] (1) -- (2) node [above, sloped,midway]{1};
  \draw [->] (1) to [bend right=35]  node [below, sloped,midway]{1} (2);
    \draw [->] (1) to [bend right=75] node [below, sloped,midway]{$3/2 K$} (2);
  \draw[snake=brace,mirror snake,raise snake=-5pt] (0,-2.5) -- (4.8,-2.5) node [below, sloped,midway]{$K$  arcs with  unit capacity~~~};  
  \pgftext[base,x=2.5cm,y=-3.35cm] {and one arc with capacity $3/2 K$} (3);
  \draw [->] (2) to [bend left=50]  node [above, sloped,midway]{$\infty$} (3);
  \draw [->] (2) to [bend left=18]  node [above, sloped,midway]{$\infty$} (3);
   \draw [->] (2) to [bend right=18]  node [above, sloped,midway]{$\infty$} (3);
  \draw [->] (2) to [bend right=50] node [below, sloped,midway]{$\infty$} (3);
  \draw[snake=brace,mirror snake,raise snake=-5pt] (5.2,-2.5) -- (10,-2.5) node [below, sloped,midway]{~~~$\Gamma+1$ arcs with infinite capacity};
\draw[loosely dotted] (7.5,-1.05) -- (7.5,-0.65);
\draw[loosely dotted] (2.5,-0.25) -- (2.5,-0.65);
\end{scope}
\end{tikzpicture}
}}
\end{center}
\caption{\label{fig:NI-RNI-RNI(P)} Illustration of the difference between the network interdiction problem and maximum adaptive flow problem. The numbers on the arcs indicate the capacities. }
\end{figure}

\section{Randomized network interdiction}
In the network interdiction problem, the interdictor goes first and determines $\Gamma$ arcs to be removed. The flow player observes the set of removed arcs and determines a flow to be sent through the remaining network. In this case, the flow player has complete knowledge of the interdictor's behavior. Our goal is to make the interdictor more powerful and make the flow player weaker. To achieve this, we allow the interdictor to use randomness to decide which strategy to play. More precisely, the interdictor assigns a probability to each pure strategy, and then randomly selects a pure strategy according to the probabilities. The flow player does not see the interdictor's strategy, but observes a probability distribution of how the interdictor decides to select arcs. 

In what follows, we formally define the randomized network interdiction problem. We first focus on the arc-based formulation of flows and then turn our attention to flows on paths, instead of arcs. We notice that the two formulations are equivalent for the network interdiction problem, but they differ for the randomized version, as shown later. Therefore, we treat the two formulations separately.


\subsection{Arc-based formulation}
A \emph{mixed} (or \emph{randomized}) strategy over $\Omega$ is given by a probability distribution $\alpha:\Omega\rightarrow [0,1]$, where $\alpha(\mu)$ is the probability that strategy $\mu$ is selected by the interdictor. 
We denote the set of all mixed strategies over $\Omega$ by $\Delta(\Omega)$.
We extend the payoff function to mixed strategies by defining
\begin{align*}
f(\alpha,x):=\sum_{\mu\in \Omega}\alpha(\mu)f(\mu,x)&& \forall \alpha \in \Delta(\Omega), ~x\in \mathcal{X}.
\end{align*}
The value $f(\alpha,x)$ represents the \emph{expected} payoff  of the game if the interdictor chooses a mixed strategy $\alpha \in \Delta(\Omega)$ and the flow player selects a pure strategy $x \in \mathcal{X}$. 

Given a mixed strategy $\alpha$, the flow player aims to find a flow with maximum expected value. The interdictor wishes to choose a mixed strategy to minimize this value. Therefore, the interdictor deals with the following problem:
\begin{align}
\label{pro:RNI}
\opt_{\text{RNI}}:=\min_{\alpha\in \Delta(\Omega)} ~\max_{x\in \mathcal{X}}~f(\alpha,x).
\end{align}
We refer to this problem as the \emph{randomized network interdiction} problem.

We next show that $\opt_{\text{RNI}}=\opt_{\text{ADP}}$. This shows that randomization permits the interdictor to perform as well as when she has perfect knowledge of the flow player's choice. To prove this, we allow the flow player to select a flow randomly. It is well known from game theory (see, e.g., \cite{Wald45}) that if both players select their strategies randomly, then there exists an equilibrium; that is, no matter which player selects her strategy first, no one has an incentive to change her mixed strategy.  Notice that the set of pure strategies for the flow player is an infinite set. Here, a mixed strategy is given by a finite distribution over $\mathcal{X}$. In fact, a \emph{random} strategy over $\mathcal{X}$ is a probability distribution $\beta: \mathcal{X}\rightarrow [0,1]$ with finite support, that is, it only assigns a non-zero value to a finite number of $s$-$t$-flows. The value $\beta(x)$ gives the probability that the flow $x$ is selected. 
We denote the set of all mixed strategies over $\mathcal{X}$ by $\Delta(\mathcal{X})$.

We extend the payoff function $f$ to mixed strategies for both players by defining
\begin{align*}
f(\alpha,\beta):=\sum_{\mu\in \Omega}\sum_{x\in \mathcal{X}: \beta(x)>0} \alpha(\mu) \beta(x) f(\mu,x)&&\forall \alpha\in \Delta(\Omega),~\beta\in\Delta(\mathcal{X}).
\end{align*}
The value $f(\alpha,\beta)$ gives the expected payoff of the game  if the interdictor chooses a mixed strategy $\alpha$ and the flow player selects a mixed strategy $\beta$. If the interdictor chooses a pure strategy $\mu$ and the flow player chooses a mixed strategy $\beta$, we denote the expected payoff by $f(\mu,\beta):=\sum_{x\in \mathcal{X}: \beta(x)>0} \beta(x) f(\mu,x)$.

\begin{theorem}
\label{thm:RNI=Adp}
$\opt_{\text{RNI}}=\opt_{\text{ADP}}$.
\end{theorem}
\begin{proof}
The result follows from the following observations:
\begin{align}
\opt_{\text{RNI}}=\min_{\alpha\in \Delta(\Omega)} ~\max_{x\in \mathcal{X}}~f(\alpha,x)
&= \min_{\alpha\in \Delta(\Omega)} ~\max_{\beta\in \Delta(\mathcal{X})}~ f(\alpha,\beta)\label{eq:RNO=Adp1}\\
&=\max_{\beta\in \Delta(\mathcal{X})}~ \min_{\alpha\in \Delta(\Omega)}~ f(\alpha,\beta)\label{eq:RNO=Adp2}\\
&=\max_{\beta\in \Delta(\mathcal{X})}~ \min_{\mu\in \Omega} ~f(\mu,\beta)\label{eq:RNO=Adp3}\\
&=\max_{x\in \mathcal{X}}~ \min_{\mu\in \Omega}~ f(\mu,x)=\opt_{\text{ADP}}\label{eq:RNO=Adp4}.
\end{align}

The second equality in Equation~\eqref{eq:RNO=Adp1} holds since the payoff function $f(\alpha,x)$ is concave in $x$ and the pure strategy set $\mathcal{X}$ is convex.

The equality in Equation~\eqref{eq:RNO=Adp2} follows from the well-known Wald's Minimax Theorem \cite{Wald45}  due to the fact that pure strategy set $\Omega$ is finite. 

Furthermore, the equality in Equation~\eqref{eq:RNO=Adp3} holds since 
\begin{align*}
\min_{\alpha\in \Delta(\Omega)} f(\alpha,\beta)=\min_{\mu\in \Omega} f(\mu,\beta)
\end{align*}
for a fixed $\beta\in \Delta(\mathcal{X})$ due to the fact that $\Omega$ is a finite set.  

It remains to prove the validity of the first equality in Equation~\eqref{eq:RNO=Adp4}. For a fixed $x\in \mathcal{X}$, we define $\arval(x):=\min_{\mu\in \Omega} f(\mu,x)$.   This function is concave in $x$. Hence, the first equality in Equation~\eqref{eq:RNO=Adp4} holds. This completes the proof of the theorem.
\end{proof}

\subsection{Path-based formulation}
So far, we have considered flows in an arc-based formulation. We next focus on an alternative formulation of flows, in which the flow player must specify paths on which to route the material. 
This leads to a different model for the randomized network interdiction problem.

Let $\mathcal{P}$ denote the set of all $s$-$t$-paths (i.e., paths from $s$ to $t$). For $P\in \mathcal{P}$, we write $e\in P$ to indicate that arc $e\in E$ lies on $P$. An $s$-$t$-(path-based) flow is a function $x:\mathcal{P}\rightarrow\rr_+$ that assigns a nonnegative value to each path so that the total flow on each arc does not exceed the capacity of the arc, that is,
\begin{align*}
\sum_{P\in \mathcal{P}:e\in P} x_{P}&\leq u_e&&\forall e\in E.
\end{align*}
The value of $x$ is the sum of the flows on the paths, i.e., $\val(x)=\sum_{P\in \mathcal{P}}x_{P}$. We use $\mathcal{X}_P$ to denote the set of all $s$-$t$-path-based flows.

Notice that the flow on path $P$ cannot reach the sink if some arc in $P$ is removed. In particular, if the interdictor selects the strategy $\mu$ and the flow player chooses a flow $x\in \mathcal{X}_P$, then the payoff of the game is given by 
\begin{align*}
g(\mu,x):=&\sum_{P\in\mathcal{P}}\max\Big\{0,1-\sum_{e\in P} \mu_e \Big\}x_P
\end{align*}
This function differs from the arc-based function $f$, defined in Equation~\eqref{pro:f(mu,x)}, because flows are not permitted in this version to be routed. The value $g(\mu,x)$ gives the amount of flow that can reach the sink if the the arcs in $\mu$ are removed.  We point out that if no arc in a path $P$ is removed, then $\max\big\{0,1-\sum_{e\in P} \mu_e \big\}=1$ and the flow on path $P$ is counted in computing the value $g(\mu,x)$. Otherwise, we will have $\max\big\{0,1-\sum_{e\in P} \mu_e \big\}=0$, and then the flow on path $P$ does not contribute to the value $g(\mu,x)$.

We present an alternative, but equivalent, formulation of the network interdiction problem as follows:
\begin{align}
  \label{pro:NI-path}
\opt_{\text{NI}}^{\text{Path}}:=\min_{\mu\in \Omega}~\max_{x\in \mathcal{X}_P} &\quad g(\mu,x).
\end{align}
We now consider the case where the flow player has to choose a flow up front before the interdictor chooses her strategy. In this situation, the flow player addresses the following problem:
\begin{align}
  \label{pro:Adp-path}
\opt_{\text{ADP}}^{\text{Path}}:=\max_{x\in \mathcal{X}_P} ~\min_{\mu\in \Omega}&\quad g(\mu,x).
\end{align} 
This problem is introduced by Aneja \emph{et al.}~\cite{AnejaNairChand01}, who study the case where only one arc is
permitted to be removed.  They show that the problem is solvable in polynomial
time in this special case. Later, Du and Chandrasekaran~\cite{DuChan07} show
that the problem is NP-hard for if two arcs can be removed. Bertsimas \emph{et al.}~\cite{BertNasrStil12} examine the problem in a general setting and propose approximation approchaes to obtain near optimal solutions. 

We next assume that the interdictor uses randomness to select arcs to be deleted. This leads to the following problem, referred to as the \emph{randomized network interdiction} problem in the \emph{path-based formulation}:
\begin{align}
\label{pro:RNI-path}
\opt_{\text{RNI}}^{\text{Path}}:=\min_{\alpha\in \Delta(\Omega)} ~\max_{x\in \mathcal{X}_P}~\sum_{\mu\in \Omega}\alpha(\mu) g(\mu,x).
\end{align}

\begin{theorem}
\label{thm:RNI=Adp-Path}
$\opt_{\text{RNI}}^{\text{Path}}=\opt_{\text{ADP}}^{\text{Path}}$.
\end{theorem}
\begin{proof}
The flow player does not benefit by choosing a mixed strategy because the payoff function $g(\mu,x)$ is linear in $x$ and the set of pure strategies $\mathcal{X}_P$ is convex. The proof now follows in a similar way as in the proof of Theorem~\ref{thm:RNI=Adp}.
\end{proof}

Problems~\eqref{pro:NI} and \eqref{pro:NI-path} are equivalent; that is, $\opt_{\text{NI}}=\opt_{\text{NI}}^{\text{Path}}$, since the interdictor first chooses arcs  to be removed and then the flow player solves a maximum flow problem in the remaining network. But this situation becomes different for the randomized version. To see the difference between the arc-based and path-based formulations, we refer to the network in Figure~\ref{fig:NI-RNI-RNI(P)}. In this network, $\opt_{\text{RNI}}^{\text{Path}}=\frac{2K}{\Gamma+1}$, while $\opt_{\text{RNI}}=\frac{5K}{2(\Gamma+1)}$. Thus, $\opt_{\text{RNI}}/\opt_{\text{RNI}}^{\text{Path}}= \frac{5}{4}$ for $\Gamma\geq 2$. We will show that this ratio is bounded by $\Gamma$ and this bound is tight. For $\Gamma=1$, we will show in the next section that the two formulations are equivalent and one can compute an optimal mixed strategy in polynomial-time.

\section{Complexity results}
In this section, we investigate computational complexity of the randomized network interdiction problem. By Theorems~\ref{thm:RNI=Adp} and \ref{thm:RNI=Adp-Path}, we know $\opt_{\text{RNI}}=\opt_{\text{ADP}}$ and $\opt_{\text{RNI}}^{\text{Path}}=\opt_{\text{ADP}}^{\text{Path}}$, respectively. Thus, complexity results for  computing $\opt_{\text{ADP}}$ and $\opt_{\text{ADP}}^{\text{Path}}$ carry over $\opt_{\text{RNI}}$ and $\opt_{\text{RNI}}^{\text{Path}}$, respectively. 

Bertsimas \emph{et al.}~\cite{BertNasrStil12} formulate the adaptive maximum flow problem as a linear optimization problem with exponentially many variables and constraints. When $\Gamma$ is fixed, the linear optimization problem has polynomial many variables and constrains, and thus can be solved in polynomial time. But, in general, they show that the adaptive maximum flow problem is strongly NP-hard by a reduction from the network interdiction problem. Thus, we have the following theorem.

\begin{theorem}
For a fixed $\Gamma$, the value $\opt_{\text{RNI}}$ can be computed in polynomial-time as a linear optimization problem. For a general $\Gamma$, it is strongly NP-hard to compute  $\opt_{\text{RNI}}$.
\end{theorem}

As mentioned already before, Du and Chandrasekaran~\cite{DuChan07} show that computing $\opt_{\text{ADP}}^{\text{Path}}$ is NP-hard even for the case that the interdictor is able to remove only two arcs. This leads to the following hardness result.
\begin{theorem}
The randomized network interdiction problem in the path-based formulation is NP-hard for any fixed $\Gamma\geq 2$.
\end{theorem}


We next show that one can compute the optimal value of the randomized network interdiction problem in both arc-based and path-based formulations in polynomial-time if the interdictor is capable to remove only one arc. 
\begin{theorem}
If $\Gamma=1$, then $\opt_{\text{RNI}}=\opt_{\text{RNI}}^{\text{Path}}$ and an optimal mixed strategy can be computed in polynomial-time. 
\end{theorem}
\begin{proof}
For the case of $\Gamma=1$, a mixed strategy for the interdictor corresponds to assign a nonnegative value to each arc, as the probability of removing that arc. More precisely, the set of mixed strategies for the interdictor is given by the unit simplex $\Delta$ of dimension $|E|$, that is,
 \begin{align*}
\Delta=\Big\{\alpha=(\alpha_e)_{e\in E}\mid \sum_{e\in E}\alpha_e=1,\alpha_e\geq 0\Big\}.
\end{align*}

Note that the flows on cycles are supposed to be zero since the flows on cycles do not contribute to the flow value. As a results, if the flow player chooses a flow $x$ (no matter in the arc-based or path-based formulation) and the interdictor decides to remove arc $e$, then the amount of flow reaching the sink will be 
\begin{align*}
\val(x)-x_e.
\end{align*}  
Hence, if the interdictor chooses a mixed strategy $\alpha\in \Delta$ and the flow player chooses a flow $x\in \mathcal{X}$, the payoff function is given by 
\begin{align*}
f(\alpha,x):=\sum_{e\in E} \alpha_e(\val(x)-x_e)=\val(x)-\sum_{e\in E}\alpha_e x_e.
\end{align*}
Therefore, Problems~\eqref{pro:RNI} and \eqref{pro:RNI-path} can be simplified as follows:
\begin{align*}
\opt_{\text{RNI}}=\opt_{\text{RNI}}^{\text{Path}}=\min_{\alpha\in \Delta}~\max_{x\in \mathcal{X}}& ~\val(x)-\sum_{e\in E}\alpha_e \cdot x_e.
\end{align*}
By considering the dual problem of the inner maximization problem, the above problem can be written as follows:
\begin{align}
  \label{pro:MinCut}
  \begin{aligned}
    \opt_{\text{RNI}}=\opt_{\text{RNI}}^{\text{Path}}=&\min && \sum_{e\in E}u_e  \rho_e\\
    &\text{s.t.}
      &&\begin{aligned}[t]
        \rho_e+\alpha_e+\pi_{v}-\pi_w		&\geq 0 		&&	\forall e=(v,w)\in E,\\
        \pi_t-\pi_s				&\geq 1,			&&\\
        \sum_{e\in E}\alpha_e&=1\\
        		\alpha_e,\rho_e		&\geq 0 		&&	\forall e\in E.
      \end{aligned}
  \end{aligned}
\end{align}
 Therefore, one can determine an optimal mixed strategy by solving this linear optimization problem. This completes the proof of the theorem.
\end{proof}
%
%


\section{On the power of randomization}
In this section, we provide tight bounds on the ratio of the optimal value of the network interdiction problem to that of randomized versions. 
In particular, our main result is the following theorem. 

\begin{theorem}
\label{thm:main}It is always true that
\begin{align}
\label{eq:NI/RNI<=T+1} \frac{\opt_\text{NI}}{\opt_\text{RNI}}&\leq \Gamma+1,\\
\label{eq:NI/RNI-path<=T+1} 
\frac{\opt_\text{NI}}{\opt_\text{RNI}^{\text{Path}}}&\leq \Gamma+1,\\
\label{eq:RNI/RNI-path<=T} 
\frac{\opt_\text{RNI}}{\opt_\text{RNI}^{\text{Path}}}&\leq \Gamma.
\end{align}
and these bounds are tight.
\end{theorem}

To prove this theorem, we require several lemmas. The core of the our analysis is based on the following parametric linear optimization problem:
\begin{align}
  \label{pro:LO}
  \begin{aligned}
    \opt_{\text{LO}}(\theta):=&\max 		&& \val(x)-\Gamma  \theta \\
    &\text{s.t.}	 &&
    \begin{aligned}[t]
                        \sum_{e\in \delta^+(v)}x_e-\sum_{e\in \delta^-(v)}x_e   &=0,             		&&\forall v\in V\setminus\{s,t\},\\
                         0\leq x_e	&\leq u_e,       		&&\forall e\in E,\\
        				 x_e	&\leq \theta, 	&&\forall e\in E.
        \end{aligned}
  \end{aligned}
\end{align}
We let $ \opt_{\text{LO}}:=\max_{\theta\geq 0}  \opt_{\text{LO}}(\theta)$ and refer to the latter problem as the LO model. This model is examined by Bertsimas \emph{et al.}~\cite{BertNasrStil12} to find approximations for Problems~\eqref{pro:Adp} and \eqref{pro:Adp-path}. They demonstrate the ability of the LO model for obtaining  good solutions in a computational study. 

We first show that the optimal value of the LO model gives a lower bound on $\opt_{\text{RNI}}^{\text{Path}}$, $\opt_{\text{RNI}}$, and $\opt_{\text{NI}}$. 

\begin{lemma}
\label{lem:LO<=NI}
We have
\begin{align}
\label{eq:LP<=RNI}
\opt_{\text{LO}}\leq \opt_{\text{RNI}}^{\text{Path}}\leq \opt_{\text{RNI}}\leq \opt_{\text{NI}}.
\end{align}
\end{lemma}
\begin{proof}
By Theorems~\ref{thm:RNI=Adp} and \ref{thm:RNI=Adp-Path}, we know $\opt_{\text{RNI}}=\opt_{\text{ADP}}$ and $\opt_{\text{RNI}}^{\text{Path}}=\opt_{\text{ADP}}^{\text{Path}}$, respectively. Thus, it suffices to show that
\begin{align*}
\opt_{\text{LO}}\leq \opt_{\text{ADP}}^{\text{Path}}\leq \opt_{\text{ADP}}\leq \opt_{\text{NI}}.
\end{align*}
Here, the first inequality from the right  is immediate since
\begin{align*}
\opt_{\text{ADP}}=\max_{x\in \mathcal{X}}~\min_{\mu\in \Omega}~f(\mu,x)&\leq\min_{\mu\in \Omega}~ \max_{x\in \mathcal{X}}~f(\mu,x)=\opt_{\text{NI}}.
\end{align*}  
The second inequality is intuitively straightforward because  the flow player in Problem~\eqref{pro:Adp-path} is more restricted than Problem~\eqref{pro:Adp}. 

Therefore, it remains to prove $\opt_{\text{LO}}\leq \opt_{\text{ADP}}^{\text{Path}}$. We assume that
the optimal value of the LO model is strictly positive since otherwise the statement is trivial. Let  $(x^*,\theta^*)$ be an optimal solution for the LO model and $(x^*_P)_{P\in\mathcal{P}}$ be an
 arbitrary path-decomposition of $x^*$.
It is sufficient to prove that 
 \begin{align}
 \label{eq:LO<=ADP-Path}
 \opt_{\text{LO}}=\val(x^*)-\Gamma \theta^*\leq ~\min_{\mu\in \Omega}&\quad g(\mu,x^*).
\end{align}
To this end, we show that there exist $\Gamma$ arcs, say
 $e_1,\ldots,e_\Gamma$, with $x^*_{e_1}=\ldots=x^*_{e_\Gamma}=\theta^*$ such
 that at most one of them lies on each path $P$ with $x^*_P>0$. 
Suppose, by contradiction, that there are at most $k$ ($k<\Gamma$)
arcs $e_1,\ldots,e_k$ with $x^*_{e_1}=\ldots=x^*_{e_k}=\theta^*$ such that
each path $P$ with $x^*_P>0$ contains at most one of these arcs. This
implies that there are $k$ flow-carrying paths $P_1,\ldots,P_k$ such
that each arc $e$ with $x^*_e=\theta^*$ lies on one of these paths. 

We now define a new solution $(x,\theta)$ for
the LO model by setting $\theta:=\theta^*-\epsilon$ and
$x_P:=x^*_P-\epsilon$, if $P=P_1,\ldots,P_k$, and $x_P:=x^*_P$,
otherwise. The $\epsilon$ is strictly positive and is chosen small
enough to ensure that $\sum_{P\in \mathcal{P}: e\in P}x_P\leq \theta$. The objective function
value of the LO model for $(x, \theta)$ is 
\begin{align*}
\sum_{P\in \mathcal{P}}x_P-\Gamma \theta
                            &=\sum_{P\in\mathcal{P} \setminus\{P_1,\ldots,P_k\}}x_P+\sum_{i=1}^{k}x_P-\Gamma  \theta\\
                            &=\sum_{P\in\mathcal{P} \setminus\{P_1,\ldots,P_k\}}x^*_P+\sum_{i=1}^{k}(x^*_P-\epsilon)-\Gamma  (\theta^*-\epsilon)\\
                            &=\sum_{P\in\mathcal{P}}x^*_P-\Gamma  \theta^* +(\Gamma-k) \epsilon> \sum_{P\in\mathcal{P}}x^*_P-\Gamma  \theta^*.
\end{align*}
This contradicts with the optimality of $(x^*, \theta^*)$, which proves the validity of Inequality \eqref{eq:LO<=ADP-Path}. This completes the proof of the lemma.
\end{proof}

In what follows, we exploit structural properties of the LO model that are needed for the proof of Theorem \ref{thm:main}. We first give some basic definitions and notation. An \emph{$s$-$t$-cut} is defined as a subset $S\subseteq V$ of nodes with $s\in S$ and $t\in V\setminus S$. The \emph{capacity} $\ca(S)$ of $S$
is defined as the sum of the capacities of the arcs going from $S$ to $V\setminus S$, that is, $\ca(S):=\sum_{e\in\delta^+(S)} u_e$. Here and subsequently, $\delta^+(S)$ denotes the set of arcs $e=(v,w)$ with $v\in S$ and $w\in V\setminus S$. We use ${\mathcal{S}}$ to denote the set of all $s$-$t$-cuts. 
For a given value $\theta\geq 0$, we let $u_e(\theta):=\min\{u_e,\theta\}$, and we let $\ca(S,\theta):=\sum_{e\in \delta^+(S)}u_e(\theta)$ denote the capacity of the cut with respect to the arc capacities $u(\theta)$.
We let $A(S,\theta)$ denote the set of all arcs $e\in \delta^+(S)$ with $\theta\leq  u_e$, and we let $B(S,\theta)$ denote the set of all arcs $e\in \delta^+(S)$ with $\theta< u_e$.

\begin{lemma}
\label{lem:LO-cut}
Suppose that $(x^*,\theta^*)$ is an optimal solution to the LO model with maximum value $\theta^*$ (i.e., if there are multiple optimal solutions, the one with the largest value $\theta^*$ is selected). 
\begin{enumerate}[(i)]
\item\label{it:LO-cut1} There exists an $s$-$t$-cut $S'$ so that 
\begin{align*}
\val(x^*)=\ca(S',\theta^*) && \text{and} && |A(S',\theta^*)|\geq \Gamma.
\end{align*} 
\item\label{it:LO-cut2} There exists an $s$-$t$-cut $S''$ so that 
\begin{align*}
\val(x^*)=\ca(S'',\theta^*)  && \text{and} &&|B(S'',\theta^*)|< \Gamma.
\end{align*} 
\end{enumerate}
\end{lemma}
\begin{proof}
For each $\epsilon>0$, we have
\begin{align}
\label{eq:LO(theta)<=LO(theta-epsilon)}
\opt_{\text{LO}}(\theta^*)&\geq \opt_{\text{LO}}(\theta^*-\epsilon),\\
\label{eq:LO(theta)<=LO(theta+epsilon)}
\opt_{\text{LO}}(\theta^*)&> \opt_{\text{LO}}(\theta^*+\epsilon),
\end{align}
since $(x^*,\theta^*)$ is an optimal solution with maximum value $\theta^*$. In addition, there exists an $s$-$t$-cut $S'$ which is a minimum cut with respect to arc capacities $u(\theta^*)$ and arc capacities $u(\theta^*-\epsilon)$ for a very small $\epsilon> 0$. More precisely, it is enough to choose $\epsilon$ as follows:
\begin{align*}
\epsilon:=\frac{1}{|E|}\min_{S\in \mathcal{S}}\{\ca(S,\theta^*)-\val(x^*)\mid \ca(S,\theta^*)-\val(x^*)>0\}.
\end{align*}
Therefore, we can write
\begin{align*}
\opt_{\text{LO}}(\theta^*)&=\ca(S',\theta^*)-\Gamma\theta^*,\\
\opt_{\text{LO}}(\theta^*-\epsilon)&=\ca(S',\theta^*-\epsilon)-\Gamma(\theta^*-\epsilon)\\
							&=\sum_{e\in S'\setminus A(S',\theta^*)}u_e+\sum_{e\in A(S',\theta^*)}(\theta^*-\epsilon)-\Gamma\theta^*+\Gamma\epsilon\\
							&=\ca(S',\theta^*)-\epsilon|A(S',\theta^*)|-\Gamma\theta^*+\Gamma\epsilon.
\end{align*}
It then follows from Inequality~\eqref{eq:LO(theta)<=LO(theta-epsilon)} that $|A(S',\theta^*)|\geq \Gamma$.

We prove the second part of the lemma by a similar argument. There exists an $s$-$t$-cut $S''$, which is a minimum cut with respect to arc capacities $u(\theta^*)$ and $u(\theta^*+\epsilon)$ for a very small $\epsilon> 0$.  Therefore,
\begin{align*}
\opt_{\text{LO}}(\theta^*)&=\ca(S'',\theta^*)-\Gamma\theta^*,\\
\opt_{\text{LO}}(\theta^*+\epsilon)&=\ca(S'',\theta^*+\epsilon)-\Gamma(\theta^*+\epsilon)\\
						   &=\ca(S'',\theta^*)+\epsilon|B(S'',\theta^*)|-\Gamma\theta^*-\Gamma\epsilon.
\end{align*}
It now follows from Inequality~\eqref{eq:LO(theta)<=LO(theta+epsilon)} that $|B(S'',\theta^*)|< \Gamma$. 
\end{proof}

\begin{lemma}
\label{lem:cases}
Suppose that $(x^*,\theta^*)$ is an optimal solution to the LO model with maximum value $\theta^*$.  Then,
\begin{enumerate}[(i)]
\item\label{NI=LOifLO<val(x)}  $\opt_\text{NI}=\opt_\text{LO}$ if $\opt_\text{LO}< \frac{1}{\Gamma+1}\val(x^*)$;
\item\label{RNI=Val(x)ifLO<delta}  $\opt_\text{RNI}=\opt_\text{LO}$ if $\opt_\text{LO}< \theta^*$;
\item\label{NI<=Val(x)} $\opt_{\text{NI}}\leq \val(x^*)$;
\item\label{RNI=Val(x)ifx=OPT}  $\opt_\text{RNI}=\opt_\text{LO}$ if $x^*$ is a maximum flow for the nominal problem;
\item\label{RNI<=Val(x)-delta}  $\opt_{\text{RNI}}\leq \val(x^*)-\theta^* $;
\end{enumerate}
\end{lemma}

\begin{proof}

\textbf{Part~\eqref{NI=LOifLO<val(x)}:}
It follows from $\opt_\text{LO}< \frac{1}{\Gamma+1}\val(x^*)$ that $\val(x^*)<(1+\Gamma)\theta^*$. In addition, by Part~\eqref{it:LO-cut1} of Lemma~\ref{lem:LO-cut}, there exists an $s$-$t$-cut $S'$ with $\val(x^*)=\ca(S',\theta^*)$ and $|A(S',\theta^*)|\geq \Gamma$.  Therefore,
\begin{align*}
(1+\Gamma)\theta^*>\val(x^*)&=\ca(S',\theta^*)=\sum_{e\in A(S',\theta^*)}\theta+\sum_{e\in \delta^+(S')\setminus A(S',\theta^*)}u_e\\
&=\theta|A(S',\theta^*)|+\sum_{e\in \delta^+(S')\setminus A(S',\theta^*)}u_e,
\end{align*}
and consequently
\begin{align*}
|A(S',\theta^*)|=\Gamma && \text{and} && \opt_{\text{LO}}=\val(x^*)-\Gamma\theta^* =\sum_{e\in \delta^+(S')\setminus A(S',\theta)}u_e.
\end{align*} 
If the arcs in $A(S',\theta^*)$ are removed, then the maximum flow value in the remanning network is at most $\sum_{e\in \delta^+(S')\setminus A(S',\theta^*)}u_e$, which is equal to $\opt_{\text{LO}}$. This implies that $ \opt_\text{NI}\leq \opt_\text{LO}$. On the other hand, it follows from Theorem~\ref{lem:LO<=NI} that $\opt_\text{LO}\leq \opt_\text{NI}$. Hence, we must have $\opt_\text{LO}=\opt_\text{NI}$.

\textbf{Part~\eqref{RNI=Val(x)ifLO<delta}:}
The inequality $\opt_\text{LO}< \theta^*$ holds if and only if $\opt_\text{LO}< \frac{1}{\Gamma+1}\val(x^*)$. Therefore, it follows from the previous part that $\opt_\text{LO}= \opt_\text{NI}$. Moreover, by Theorem~\ref{lem:LO<=NI}, we have $\opt_\text{LO}\leq \opt_\text{RNI}\leq \opt_\text{NI}$. This implies that $\opt_\text{LO}=\opt_\text{RNI}$.

\textbf{Part~\eqref{NI<=Val(x)}:} By Part~\eqref{it:LO-cut2} of Lemma~\ref{lem:LO-cut}, there exists an $s$-$t$-cut $S''$ with $\val(x^*)=\ca(S'',\theta)$ so that $|B(S'',\theta^*)|< \Gamma$. For each $e\in \delta^+(S'')\setminus B(S'',\theta^*)$, we have $u_e\leq \theta^*$.  Therefore, if the arcs in $B(S'',\theta^*)$ are removed, the maximum amount of flow that can be sent from $s$ to $t$ is at most $\sum_{e\in \delta^+(S'')\setminus B(S'',\theta^*)}u_e$. This means that $ \opt_{\text{NI}}\leq \sum_{e\in \delta^+(S'')\setminus B(S'',\theta^*)}u_e$. Hence, we can write
\begin{align}
\label{eq:NI<=Val(x)-theta|B|}
\val(x^*)=\ca(S'',\theta)&=\sum_{e\in B(S'',\theta^*)}\theta+\sum_{e\in \delta^+(S'')\setminus B(S'',\theta^*)}u_e
\\&\geq \theta|B(S'',\theta^*)|+ \opt_{\text{NI}},
\end{align}
and consequently $\opt_{\text{NI}}\leq \val(x^*)$.

\textbf{Part~\eqref{RNI=Val(x)ifx=OPT}:}  For an $s$-$t$-cut $S$ and an $s$-$t$-flow $x$, we define 
\begin{align*}
R(x,S):=\min_{\mu\in \Omega}\sum_{e\in \delta^+(S)}(1-\mu_e)  x_e.
\end{align*}
It follows from Lemma 8 in \cite{BertNasrStil12} that 
\begin{align*}
\min_{\mu \in \Omega} f(\mu,x)=\min_{S\in \mathcal{S}} R(x,S).
\end{align*}
 Therefore, we can write
\begin{align*}
\opt_\text{ADP}= \max_{x\in \mathcal{X}}~\min_{S\in \mathcal{S}}~R(x,S)& \leq  \min_{S\in \mathcal{S}}~\max_{x\in \mathcal{X}}~R(x,S).
\end{align*}
Furthermore, we know from Part~\eqref{it:LO-cut1} of Lemma~\ref{lem:LO-cut} that there exists a cut $S'$ with $|A(S',\theta^*)|\geq \Gamma$ so that
\begin{align*}
\val(x^*)=\ca(S',\theta^*)=\sum_{e\in A(S',\theta^*)}\theta^*+\sum_{e\in \delta^+(S')\setminus A(S',\theta)}u_e.
\end{align*}
Therefore, we can write
\begin{align*}
\opt_\text{ADP}\leq \max_{x\in \mathcal{X}} R(x,S')&=R(x^*,S')=\min_{\mu\in \Omega}\sum_{e\in \delta^+(S')}(1-\mu_e)  x^{*}_e\\
&=\val(x^*)-\Gamma\theta^*= \opt_\text{LO}.
\end{align*}
Moreover, we have $\opt_\text{LO}\leq \opt_\text{ADP}$ by Lemma~\ref{lem:LO<=NI}. This proves $\opt_\text{LO}=\opt_\text{ADP}$. In addition, we know $\opt_\text{ADP}=\opt_\text{RNI}$ because of Theorem~\ref{thm:RNI=Adp}. Hence, we must have $\opt_\text{LO}=\opt_\text{RNI}$.

\textbf{Part~\eqref{RNI<=Val(x)-delta}:}  If $x^*$ is a maximum flow for the nominal problem, then it follows from previous part  that $\opt_\text{RNI}=\opt_\text{LO}=\val(x^*)-\Gamma\theta\leq  \val(x^*)-\theta^* $ and we are done.  Hence, we assume that $x^*$ is not a maximum flow. Then, there exists an $s$-$t$-cut $S''$ with $\val(x^*)=\ca(S'',\theta^*)$ so that $1\leq |B(S'',\theta^*)|$ since otherwise $\val(x^*)=\ca(S'')$ and $x^*$ must be a maximum flow. Furthermore, it follows from Inequality \eqref{eq:NI<=Val(x)-theta|B|} that $\opt_{\text{NI}}\leq \val(x^*)-\theta|B(S'',\theta^*)|$. This shows that $\opt_{\text{NI}}\leq \val(x^*)-\theta^*$ since $ |B(S'',\theta^*)|\ge 1$.
\end{proof}


\begin{lemma}
\label{lem:main}
It is always true that
\begin{align}
\label{eq:NI/LO<=T+1} \frac{\opt_\text{NI}}{\opt_\text{LO}}&\leq \Gamma+1,\\
\label{eq:RNI/LO<=T} \frac{\opt_\text{RNI}}{\opt_\text{LO}}&\leq \Gamma.
\end{align}
\end{lemma}

\begin{proof}
Suppose that $(x^*,\theta^*)$ is an optimal to the LO model with maximum value $\theta^*$. 
If $\opt_\text{LO}\geq \frac{1}{\Gamma+1}\val(x^*)$, then $\opt_\text{LO}\geq \frac{1}{\Gamma+1}\opt_{\text{NI}}$ because of Part~\eqref{NI<=Val(x)} of Lemma~\ref{lem:cases}, and consequently Inequality~\eqref{eq:NI/LO<=T+1} holds.  
Hence, we assume that $\opt_\text{LO}< \frac{1}{\Gamma+1}\val(x^*)$. This implies that $\opt_\text{LO}<\theta^* $. Therefore, $\opt_\text{LO}=\opt_\text{NI}$ because of Part~\eqref{RNI=Val(x)ifLO<delta} of Lemma~\ref{lem:cases}. This establishes Inequality~\eqref{eq:NI/LO<=T+1}. This also shows that Inequality~\eqref{eq:NI/RNI<=T+1} is true since $\opt_\text{LO}\leq \opt_\text{RNI}\leq \opt_\text{RNI}$ by Theorem~\ref{lem:LO<=NI}. 

We proceed to prove the validity of Inequality~\eqref{eq:RNI/LO<=T}.  If $\opt_\text{LO}<\theta^*$, then by Part~\eqref{RNI=Val(x)ifLO<delta} of Lemma~\ref{lem:cases} we must have $\opt_{\text{RNI}}=\opt_\text{LO}$, and consequently Inequality~\eqref{eq:RNI/LO<=T} holds. Thus, we assume that $\opt_\text{LO}\geq\theta^*$.  We can write
\begin{align*}
\opt_{\text{RNI}}\leq \val(x^*)-\theta^* =\opt_\text{LO}+(\Gamma-1)\theta^*\leq  \Gamma\cdot \opt_\text{LO},
\end{align*}
where the first inequality follows from Part~\eqref{NI<=Val(x)} of Lemma~\ref{lem:cases} and the second inequality follows from the fact that   
$\opt_\text{LO}\geq\theta^*$. This shows that Inequality~\eqref{eq:RNI/LO<=T} always holds. 
\end{proof}


{\emph{Proof of Theorem \ref{thm:main}}.} The validity of the bounds in \eqref{eq:NI/RNI<=T+1}, \eqref{eq:NI/RNI-path<=T+1}, and \eqref{eq:RNI/RNI-path<=T} immediately follows from Lemmas \ref{lem:LO<=NI} and \ref{lem:main}.  We next provide two examples to show these bounds are all tight.

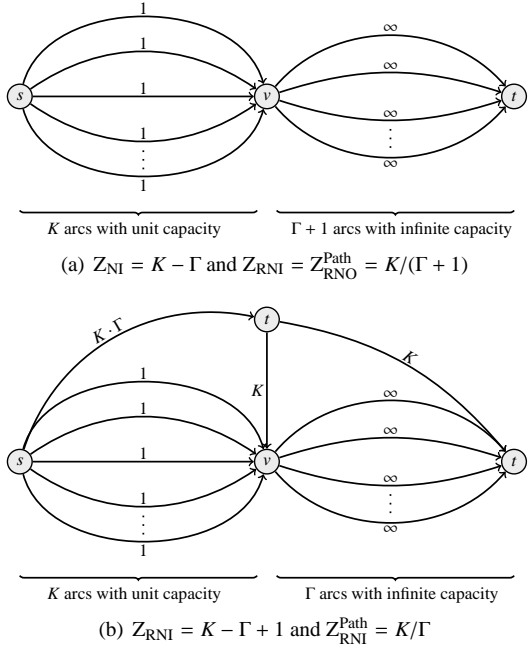
\begin{figure}[t]
    \centering
    \subfigure[\label{fig:RVal-NInt(a)} $\opt_{\text{NI}}=K-\Gamma$ and $\opt_{\text{RNI}}=\opt_{\text{RNO}}^{\text{Path}}=K/(\Gamma+1)$]
{\scalebox{0.65}{
\begin{tikzpicture}[inner sep=0.4mm]
  \node (1) at (0,0) [place] {$s$};
  \node (2) at (5,0) [place] {$v$};
  \node (3) at (10,0) [place] {$t$};
\begin{scope}[color=black,line width=1pt]
  \draw [->] (1) to [bend left=75]  node [above, sloped,midway]{1} (2);
   \draw [->] (1) to [bend left=35] node [above, sloped,midway]{1} (2);
 \draw [->] (1) -- (2) node [above, sloped,midway]{1} (2);
  \draw [->] (1) to [bend right=35] node [above, sloped,midway]{1} (2);
    \draw [->] (1) to [bend right=75] node [below, sloped,midway]{1}  (2);
  \draw[snake=brace,mirror snake,raise snake=-5pt] (0,-2.5) -- (4.8,-2.5) node [below, sloped,midway]{$K$  arcs with  unit capacity~~~};  
  \draw [->] (2) to [bend left=50]  node [above, sloped,midway]{$\infty$} (3);
  \draw [->] (2) to [bend left=18]  node [above, sloped,midway]{$\infty$} (3);
   \draw [->] (2) to [bend right=18]  node [above, sloped,midway]{$\infty$} (3);
  \draw [->] (2) to [bend right=50] node [below, sloped,midway]{$\infty$} (3);
  \draw[snake=brace,mirror snake,raise snake=-5pt] (5.2,-2.5) -- (10,-2.5) node [below, sloped,midway]{~~~$\Gamma+1$ arcs with infinite capacity};
\draw[loosely dotted] (7.5,-1.05) -- (7.5,-0.65);
\draw[loosely dotted] (2.5,-1.1) -- (2.5,-1.5);
\end{scope}
\end{tikzpicture}}}
   \subfigure[\label{fig:RVal-NInt(b)} $\opt_{\text{RNI}}=K-\Gamma+1$ and $\opt_{\text{RNI}}^{\text{Path}}=K/\Gamma$]
{\scalebox{0.65}{
\begin{tikzpicture}[inner sep=0.4mm]
  \node (1) at (0,0) [place] {$s$};
  \node (2) at (5,0) [place] {$v$};
  \node (3) at (10,0) [place] {$t$};
  \node (4) at (5,2.9) [place] {$t$};
\begin{scope}[color=black,line width=1pt]
  \draw [->] (1) to [bend left=75]  node [above, sloped,midway]{1} (2);
   \draw [->] (1) to [bend left=35] node [above, sloped,midway]{1} (2);
 \draw [->] (1) -- (2) node [above, sloped,midway]{1};
  \draw [->] (1) to [bend right=35] node [above, sloped,midway]{1} (2);
    \draw [->] (1) to [bend right=75] node [below, sloped,midway]{1} (2);
  \draw[snake=brace,mirror snake,raise snake=-5pt] (0,-2.5) -- (4.8,-2.5) node [below, sloped,midway]{$K$  arcs with  unit capacity~~~};  
  \draw [->] (2) to [bend left=50]  node [above, sloped,midway]{$\infty$} (3);
  \draw [->] (2) to [bend left=18]  node [above, sloped,midway]{$\infty$} (3);
   \draw [->] (2) to [bend right=18]  node [above, sloped,midway]{$\infty$} (3);
  \draw [->] (2) to [bend right=50] node [below, sloped,midway]{$\infty$} (3);
    \draw [->] (4) -- (2) node [left,midway]{$K$};
  \draw [->] (1) to [bend left=42]  node [above, sloped,midway]{$K\cdot \Gamma$} (4);
  \draw [->] (4) to [bend left=20] node [above, sloped,midway]{$K$} (3);
  \draw[snake=brace,mirror snake,raise snake=-5pt] (5.2,-2.5) -- (10,-2.5) node [below, sloped,midway]{~~~$\Gamma$ arcs with infinite capacity};
\draw[loosely dotted] (7.5,-1.05) -- (7.5,-0.65);
\draw[loosely dotted] (2.5,-1.1) -- (2.5,-1.5);
\end{scope}
\end{tikzpicture}}}
\caption{\label{fig:RVal-NInt} Networks for the proof of Theorem~\ref{thm:main}.  The numbers on the arcs indicate the capacities. }
\end{figure}

In the first example, we consider a network with three nodes $s$, $v$, and $t$, and parallel arcs from $s$ to $v$ and $v$ to $t$. There are $K$ parallel arcs with unit capacity from $s$ to $v$ and $\Gamma+1$ parallel arcs with infinite capacity from $v$ to $t$ (see the network in Figure \ref{fig:RVal-NInt(a)}). We let $K\geq \Gamma+1$. In this network, we have $\opt_{\text{NI}}=K-\Gamma$, whereas $\opt_{\text{RNI}}=\opt_{\text{RNO}}^{\text{Path}}=K/(\Gamma+1)$. Therefore, $\frac{\opt_{\text{NI}}}{\opt_{\text{RNI}}}=\frac{\opt_{\text{NI}}}{\opt_{\text{RNO}}^{\text{Path}}}=\frac{\Gamma+1(K -\Gamma)}{K }$.  When $K $ gets enough large, the bound becomes enough close to  $\Gamma+1$. This shows the bounds in Inequalities \eqref{eq:NI/LO<=T+1} and \eqref{eq:NI/RNI-path<=T+1} are tight.

In the second example, we consider a network with four nodes $s$, $v$, $w$, and $t$ as shown in Figure \ref{fig:RVal-NInt(b)}. There are $K$ parallel arcs with unit capacity from $s$ to $v$ and $\Gamma$ parallel arcs with infinite capacity from $v$ to $t$. In addition, there is one arc from $s$ to $v$ with capacity $\Gamma\cdot K$, one arc from $w$ to $v$ with capacity $K$, and one arc from $w$ to $t$ with infinite capacity. In this network, we have $\opt_{\text{RNI}}=K-\Gamma+1$, whereas $\opt_{\text{RNI}}^{\text{Path}}=K/\Gamma$. Therefore, $\frac{\opt_{\text{RNI}}}{\opt_{\text{RNI}}^{\text{Path}}}=\frac{\Gamma(K -\Gamma+1)}{K }$.  When $K $ gets enough large, the bound becomes enough close to $\Gamma$. This shows the bound in Inequality \eqref{eq:RNI/RNI-path<=T} is tight.
\qed

\section{Approximation bounds}
It follows from Lemma~\ref{lem:main} that the optimal value of the LO model is a $(\Gamma+1)$-approximation for $\opt_{\text{NI}}$ and a $\Gamma$-approximation for $\opt_{\text{RNI}}$. We next show that the optimal value of the LO model also provides a good approximation for $\opt_{\text{RNI}}^{\text{Path}}$. 

\begin{theorem}
We have
\begin{align}
\label{eq:RNI-path/LO<=T/4}
\frac{\opt_\text{RNI}^{\text{Path}}}{\opt_\text{LO}}&\leq 1+\frac{\lfloor \Gamma/2\rfloor \cdot \lceil \Gamma/2\rceil}{\Gamma+1},
\end{align}
and this bound is tight.
\end{theorem}

\begin{proof}
Suppose that $(x^*,\theta^*)$ is an optimal to the LO model with maximum value $\theta^*$. By Parts~\eqref{it:LO-cut1} and Parts~\eqref{it:LO-cut2} of Lemma~\ref{lem:LO-cut}, there are $s$-$t$-cuts $S'$ and $S''$ so that $|A(S',\theta^*)|\geq \Gamma$, $|B(S'',\theta^*)|< \Gamma$, and
\begin{align}
\label{eq:val(x)=cap(A)}\val(x^*)&=\sum_{e\in \delta^+(S')}u_e(\theta^*)=\sum_{e\in A(S',\theta^*)}\theta^*+\sum_{e\in \delta^+(S')\setminus A(S',\theta^*)}u_e,\\
\label{eq:val(x)=cap(B)}
\val(x^*)&=\sum_{e\in \delta^+(S'')}u_e(\theta^*)=\sum_{e\in B(S'',\theta^*)}\theta^*+\sum_{e\in \delta^+(S'')\setminus B(S'',\theta^*)}u_e.
\end{align}
Let $a:=|A(S',\theta^*)|$, $b:=|B(S'',\theta^*)|$, and $L:=\sum_{e\in \delta^+(S')\setminus A(S',\theta^*)}u_e$. Then, it follows from \eqref{eq:val(x)=cap(A)} and \eqref{eq:val(x)=cap(B)} that 
\begin{align*}
\sum_{e\in \delta^+(S'')\setminus B(S'',\theta^*)}u_e&=\val(x^*)-b\theta^*=L+(a-b)\theta^*.
\end{align*}

Now suppose that the interdictor is restricted to delete the $b$ arcs in $B(S'',\theta^*)$ and the remanning $\Gamma-b$ arcs in $A(S',\theta^*)$. After delineating the $b$ arcs in $B(S'',\theta^*)$, $L+(a-b)\theta^*$ units of flow can be pushed through the cut $S''$, and thus, can reach the cut $S'$. Then, in the best case, the flow player can send $L$ units of flow through the arcs in $\delta^+(S')\setminus A(S',\theta^*)$ and the remanning  $(a-b)\theta^*$ units of flow through the  arcs in $A(S',\theta^*)$. Since the interdictor will remove $\Gamma-b$ arcs in $A(S',\theta^*)$ due to our assumption, the flow player can send at most $\frac{(a-\Gamma+b)(a-b)\theta^*}{a}$ units of flow through the arcs arcs in $A(S',\theta^*)$ in a path-based formulation. In total, the flow player can send at $L+\frac{(a-\Gamma+b)(a-b)\theta^*}{a}$ from the source $s$ to the sink $t$ if the interdictor removes the $b$ arcs in $B(S'',\theta^*)$ and the remanning $\Gamma-b$ arcs in $A(S',\theta^*)$. This implies that
\begin{align*}
\opt_\text{RNI}^{\text{Path}}\leq L+\frac{(a-\Gamma+b)(a-b)\theta^*}{a}.
\end{align*}

On the other hand, we have 
\begin{align*}
\opt_\text{LO}=\val(x^*)-\Gamma \theta^*=L+(a-\Gamma)\theta^*.
\end{align*}
 Therefore, 
\begin{align}
\label{eq:L+(a-T)}
\frac{\opt_\text{RNI}^{\text{Path}}}{\opt_\text{LO}}\leq \frac{L+\frac{(a-\Gamma+b)(a-b)\theta^*}{a}}{L+(a-\Gamma)\theta^*}\leq \frac{(a-\Gamma+b)(a-b)}{(a-\Gamma)a}.
\end{align}

Notice that if $a=\Gamma$, then $\opt_\text{LO}=L$ and $\opt_\text{RNI}^{\text{Path}}\leq L$, and consequently Inequality~\eqref{eq:RNI-path/LO<=T/4} holds. Hence, we assume that $a\geq \Gamma+1$. In this case, the right hand side of Inequality \eqref{eq:L+(a-T)} attaints its maximum when $a=\Gamma+1$ and $b=\lfloor \Gamma/2\rfloor$.  By substituting $a=\Gamma+1$ and $b=\lfloor \Gamma/2\rfloor$ in the the right hand side of Inequality \eqref{eq:L+(a-T)} , we obtain
\begin{align*}
\frac{\opt_\text{RNI}^{\text{Path}}}{\opt_\text{LO}}\leq 1+\frac{\lfloor \Gamma/2\rfloor \cdot \lceil \Gamma/2\rceil}{\Gamma+1}.
\end{align*}
This establishes the validity of Inequality~\eqref{eq:RNI-path/LO<=T/4}. We next show that this bound is tight.

We consider a network with three nodes $s$, $v$, and $t$. There are $\Gamma$ parallel arcs with unit capacity from $s$ to $v$ and $\Gamma+1$ parallel arcs with infinite capacity from $v$ to $t$. In addition, there are $\lfloor \Gamma/2\rfloor$ parallel arcs from $s$ to $v$ with capacity $K$. In this network, we have $\opt_{\text{RNO}}^{\text{Path}}=\frac{(\lfloor \Gamma/2\rfloor+1)K}{\Gamma+1}$, whereas $\opt_{\text{LO}}=\frac{K}{\lceil \Gamma/2\rceil+1}$. Therefore, 
\begin{align*}
\frac{\opt_{\text{RNO}}^{\text{Path}}}{\opt_{\text{LO}}}=\frac{(\lfloor \Gamma/2\rfloor+1)\cdot \lceil \Gamma/2\rceil+1}{\Gamma+1}=1+\frac{\lfloor \Gamma/2\rfloor \cdot \lceil \Gamma/2\rceil}{\Gamma+1}.
\end{align*}  
This shows that the bound in Inequality \eqref{eq:RNI-path/LO<=T/4} is tight.
\end{proof}

\bibliographystyle{ormsv080}
\bibliography{mybib}

\end{document}